\documentclass{article}

\usepackage{amsthm}
\usepackage{hyperref}
\usepackage{mathtools}
\usepackage{tikz-cd}
\usepackage{enumitem}

\newcommand{\Lin}{\mathrm{L}}

\newcommand{\tr}{\mathrm{tr}}
\newcommand{\id}{\mathrm{id}}
\newcommand{\linspan}{\mathrm{span}}

\newtheorem{proposition}{Proposition}
\newtheorem{remark}[proposition]{Remark}
\newtheorem{corollary}[proposition]{Corollary}
\newtheorem{lemma}[proposition]{Lemma}
\newtheorem{definition}[proposition]{Definition}

\title{Strassen $2\times2$ Matrix Multiplication from a 3-dimensional
Volume Form}
\author{Benoit Jacob\\AMD\\{\small\tt benoit.jacob@amd.com}}
\date{}

\begin{document}
\maketitle

\begin{abstract}
  The Strassen $2\times2$ matrix multiplication algorithm arises from
  the volume form on the 3-dimensional quotient space of the $2\times
  2$ matrices by the multiples of identity.
\end{abstract}

\section{Introduction}

Strassen's $2\times 2$ matrix multiplication
algorithm~\cite{Strassen} is a formula for multiplying
$2\times 2$ matrices $a$ and $b$:
\begin{equation}
  \label{original_strassen}
  ab = \tr(a)\tr(b)I + \sum_{i=1}^6 \tr(X_ia)\tr(Y_ib)Z_i
\end{equation}
where $I$ is the identity matrix, $\tr$ is the trace, and the $X_i,
Y_i, Z_i$ are constant matrices. This formula is a rank 7
decomposition of the matrix multiplication tensor, that is, a
decomposition of matrix multiplication into a sum of 7 simple tensors.

This may be applied recursively to multiply $n\times n$ matrices in
$O(n^{\log 7/\log 2})$ time, approximately $O(n^{2.81})$, opening
a research field to which the book
\cite{AlgebraicComplexityTheory} provides an introduction. One line
of research has focused on further improving this asymptotic
complexity, notably \cite{StrassenLaser}, \cite{CoppersmithWinograd}
achieving $O(n^{2.376})$ and several refinements, recently
\cite{williams2023newboundsmatrixmultiplication} and
\cite{alman2024asymmetryyieldsfastermatrix}, approaching
$O(n^{2.37})$. Another line of research has pursued decompositions
of matrix multiplication tensors for other small matrix sizes such as
$3\times 3$, $4\times 4$, etc. These have often involved numerical
searches, such as the recent \cite{novikov2025alphaevolvecodingagentscientific}.

Despite these advances, algorithms faster than
$O(n^{2.81})$ are ``almost never implemented"
\cite{IkenmeyerLysikov}, and practical evaluations such as
\cite{dalberto2023strassensmatrixmultiplicationalgorithm} have
continued favoring the Strassen $2\times 2$ algorithm. Known
algorithms for other small matrix sizes struggle to significantly
improve on it: the up-to-date table \cite{TableFastMatmul} shows
complexity exponents clustering around 2.8. For $4\times 4$ matrix
multiplication, the Strassen
algorithm has tensor rank $7^2=49$, and that remained the state of
the art for over 55 years until
\cite{novikov2025alphaevolvecodingagentscientific} and
\cite{dumas2025noncommutativealgorithmmultiplying4x4} lowered that
from 49 to 48,
achieving a complexity
exponent of $2.79$. Moreover, known
algorithms with substantially lower asymptotic complexity tend to
have large constants in the $O$, as discussed in
\cite{alman2024improvingleadingconstantmatrix}.

The Strassen $2\times 2$ algorithm also stands out from a theoretical
perspective: its tensor rank 7 is known to be optimal
\cite{WINOGRAD1971381} and it is known to be essentially unique under
that constraint \cite{DEGROOTE19781}. By contrast, the tensor rank of
$n\times n$ matrix multiplication is still unknown for all $n\geq 3$.
For $n=3$, it is still only known to be between 19 and 23, see \cite{Blaeser03}.

Optimality and uniqueness make the Strassen $2\times 2$ algorithm a
basic fact of 2-dimensional linear algebra. Such facts are expected
to be simple and geometric. However, the original statement and proof
of the Strassen algorithm are calculations on matrix coefficients.
This has motivated a quest for geometric interpretations. The recent
\cite{ChiantiniIkenmeyerLandsbergOttaviani} and
\cite{IkenmeyerLysikov} in particular were inspirational to the
present article, and \cite{IkenmeyerLysikov} contains a survey of
this endeavour, tracing it back
to the years following the publication of the original Strassen
article \cite{Strassen}. Other recent articles in this line of
research include
\cite{ikenmeyer2025strassensalgorithmorbitflip},
\cite{Burichenko}, \cite{GrochowMooreGroupOrbits} and \cite{GrochowMoore}.

The present article offers a geometric interpretation of the Strassen
algorithm by addressing a more general question: is the Strassen
algorithm an independent fact in multilinear algebra, or could it be
related to a known fact? We derive it from the expansion of a
3-dimensional volume form into an antisymmetrized sum of $3!=6$
simple tensors. That expansion follows from the one-dimensionality of
the space of antisymmetric $n$-forms, which
is an abstract version of
\href{https://en.wikipedia.org/wiki/Cavalieri\%27s_principle}{Cavalieri's
principle}, the idea that the volume of a solid is unchanged by
sliding parallel slices. As to the question of why specifically
$2\times 2$ matrices, the answer is that as matrix multiplication is
a tensor of order 3 on matrix spaces, interpreting it as a volume
form requires a 3-dimensional matrix space, and the specific case of
$2\times 2$ matrices gives us such a 3-dimensional matrix space by
taking the quotient by multiples of the identity matrix: $3=2^2-1$.

\textbf{Acknowledgements.} The author would like to thank Eugene Ha, Paolo
d'Alberto and Zach Garvey for helpful comments.

\section{Overview}

Fix, for this entire article, a 2-dimensional vector space $V$ over a
field $k$. Let $\Lin(V)$ denote the space
of linear maps from $V$ to itself. Start by considering this
trilinear form $g$ on $\Lin(V)$:
\begin{equation}
  \label{intro_g}
  g(a_1, a_2, a_3) = \tr(a_1a_2a_3) - \tr(a_3a_2a_1).
\end{equation}
We notice (Lemma \ref{g_volume_form}) that $g$ is a
volume form on the quotient of $\Lin(V)$ by the multiples of the
identity matrix, which has dimension 3. This gives (Lemma
\ref{prop_decomp_g}) rank 6 decompositions of $g$ parametrized by
bases of the dual space. Our next step is to relate $g$ to this other
trilinear form $h$ on $\Lin(V)$:
\begin{equation}
  \label{intro_h}
  h(a_1, a_2, a_3) = \tr(a_1)\tr(a_2)\tr(a_3) - \tr(a_1a_2a_3).
\end{equation}
Using the natural isomorphism $\Lin(V)^{*\otimes3} \simeq
\Lin(V^{\otimes 3})$, view the trilinear forms
$\tr(a_1)\tr(a_2)\tr(a_3)$, $\tr(a_1a_2a_3)$ and $\tr(a_3a_2a_1)$ as
respectively the permutations $\id$, (1\,2\,3) and (3\,2\,1) permuting the
terms in $V^{\otimes3}$ (Lemma \ref{lemma_t_sigma_star}). This allows
viewing $h$ as the composition of $g$ with a linear map induced by the
permutation (3\,2\,1)
(Lemma \ref{relation_g_h}), which allows transporting certain rank 6
decompositions
of $g$ into rank 6 decompositions of $h$ (Proposition
\ref{prop_decomp_h}, our main result), yielding (Corollary
\ref{concrete_decomp_h})
\begin{equation}
  \label{intro_strassen}
  \tr(a_1a_2a_3) = \tr(a_1)\tr(a_2)\tr(a_3) - \textrm{\{rank 6
  decomposition of $h$\}},
\end{equation}
which is a rank 7 decomposition of $\tr(a_1a_2a_3)$. Dualizing that
yields a rank 7 decomposition of matrix multiplication (Corollary
\ref{concrete_decomp_h_bilinear}) parametrized by a choice of basis. A
specific choice yields the original Strassen algorithm (Corollary
\ref{good_old_strassen}).

\section{Terminology and lemmas in tensor algebra}

Throughout this article, ``vector space'' means finite-dimensional
vector space. For any vector spaces $U$ and $W$ over a field $k$, let
$\Lin(U,W)$ denote the space of linear maps from $U$ to $W$. In the
case $W=U$, we write $\Lin(U)$ for $\Lin(U, U)$. In the case $W=k$,
we let $U^*=\Lin(U,k)$ denote the dual space of $U$.

Given any vector spaces $U_1,\ldots,U_n$, $W_1,\ldots,W_n$, we will
make the identification
$$\Lin(U_1, W_1)\otimes\ldots\otimes\Lin(U_n, W_n) \simeq
\Lin(U_1\otimes\ldots\otimes U_n, W_1\otimes\ldots\otimes W_n).$$
As special cases of that, for any vector space $U$ over $k$, for any
positive integer $n$,
we identify
$\Lin(U)^{\otimes n}\simeq \Lin(U^{\otimes n})$ and $U^{*\otimes n}
\simeq (U^{\otimes n})^*$. The latter identification means concretely
that given linear forms $\mu_1,\ldots,\mu_n$ on a vector space $U$,
we identify the tensor $\mu_1\otimes\ldots\otimes\mu_n$ as the
$n$-linear form on $U$ given, for all vectors $u_1,\ldots,u_n$ in $U$, by:
$$(\mu_1\otimes\ldots\otimes\mu_n)(u_1, \ldots, u_n) =
\mu_1(u_1)\ldots \mu_n(u_n).$$

Let us now describe a few other natural isomorphisms of tensor spaces
that we will keep as named isomorphisms, refraining from making identifications.
\begin{definition}
  \label{definition_iota_star}
  For any vector space $U$, define linear maps $\iota$, $\iota^*$ and $*$:
  \begin{align*}
    \iota \; &\colon \; U\otimes U^* \rightarrow \Lin(U), &v\otimes\lambda
    &\mapsto  \iota(v\otimes\lambda) = (u \mapsto \lambda(u)v) \\
    \iota^* \; &\colon \; U\otimes U^* \rightarrow \Lin(U)^*,
    &v\otimes\lambda &\mapsto  \iota^*(v\otimes\lambda) =  (a \mapsto
    \lambda(a(v))) \\
    * \; &\colon \; \Lin(U) \rightarrow \Lin(U)^*, &a &\mapsto  a^* =  (b
    \mapsto \tr(ab)).
  \end{align*}
\end{definition}

\begin{lemma}
  The linear maps $\iota$, $\iota^*$ and $*$ are isomorphisms.
\end{lemma}
\begin{proof}
  These maps are injective, and when $U$ has dimension $n$, the
  source and destination spaces have the same dimension $n^2$.
\end{proof}

\begin{lemma}
  \label{iota_equations}
  For any vector space $U$, for any $u, v$ in $U$ and any
  $\lambda,\mu$ in $U^*$, we have
  \begin{align}
    \iota(v\otimes\lambda) \iota(u\otimes\mu) &= \lambda(u)
    \iota(v\otimes\mu),\label{iota_mul}\\
    \iota^*(v\otimes\lambda)(\iota(u\otimes\mu)) &=
    \lambda(u)\mu(v),\label{iota_pairing}\\
    \tr(\iota(v\otimes\lambda)) &= \lambda(v). \label{iota_trace}
  \end{align}
\end{lemma}
\begin{proof}
  For any $w$ in $U$, we have $\iota(v\otimes\lambda)
  \iota(u\otimes\mu) (w) = \iota(v\otimes\lambda)(\mu(w) u) =
  \lambda(u) \mu(w) v = \lambda(u) \iota(v\otimes\mu) (w)$,
  establishing Equation (\ref{iota_mul}). We have
  $\iota^*(v\otimes\lambda)(\iota(u\otimes\mu)) =
  \lambda(\iota(u\otimes\mu)(v)) =
  \lambda(\mu(v)u)=\mu(v)\lambda(u)$, establishing Equation
  (\ref{iota_pairing}). Let $w_1,\ldots,w_n$ be a basis of $U$ such
  that $w_1=v$. In that basis, the matrix of $\iota(v\otimes\lambda)$ is
  $\left(
    \begin{smallmatrix}
      \lambda(v) & \lambda(w_2) & \cdots & \lambda(w_n) \\
      0 & 0 & \cdots & 0 \\
      \cdots & \cdots & \cdots & \cdots
  \end{smallmatrix}\right),$
  whose trace is $\lambda(v)$, establishing Equation (\ref{iota_trace}).
\end{proof}

\begin{lemma}
  \label{diagram_iota_star_commutes}
  For any vector space $U$, the following diagram commutes,
  justifying the notation $\iota^*$.
  \begin{equation}
    \label{diagram_iota_star}
    \begin{tikzcd}
      & U\otimes U^* \arrow[dl, "\iota" swap] \arrow[dr, "\iota^*"] & \\
      \Lin(U) \arrow[rr, "*"] & & \Lin(U)^*
    \end{tikzcd}
  \end{equation}
\end{lemma}
\begin{proof}
  The claim is that for all $v\otimes\lambda$ in $U\otimes U^*$, we
  have $\iota^*(v\otimes\lambda) = \iota(v\otimes\lambda)^*$ as
  elements of $\Lin(U)^*$. By linearity, it is enough to check that
  these forms in $\Lin(U)^*$ agree on rank one elements of $\Lin(U)$,
  which are the $\iota(u \otimes \mu)$ with $u$ in
  $U$ and $\mu$ in $U^*$. Indeed, the equations from Lemma
  \ref{iota_equations} give:
  \begin{align*}
    \iota^*(v\otimes\lambda)(\iota(u\otimes\mu)) & = \lambda(u)\mu(v)
    & \text{by Equation (\ref{iota_pairing})} \\
    & = \tr(\lambda(u) \iota(v\otimes \mu)) & \text{by Equation
    (\ref{iota_trace})}\\
    & = \tr(\iota(v\otimes\lambda) \iota(u\otimes\mu)) & \text{by
    Equation (\ref{iota_mul})}\\
    & = \iota(v\otimes\lambda)^*(\iota(u\otimes\mu)) & \text{by
    Definition \ref{definition_iota_star}.}
  \end{align*}
\end{proof}

\begin{definition}
  \label{definition_left_right_mul}
  For any vector space $U$ and any element $a$ of $\Lin(U)$, let
  $L_a, R_a$ denote respectively the left and right
  multiplication-by-$a$ maps: $L_a(b)=ab$ and $R_a(b) = ba$ for all
  $b$ in $\Lin(U)$. Thus $L_a$ and $R_a$ are elements of $\Lin(\Lin(U))$.
\end{definition}

\begin{lemma}
  \label{star_product}
  For any vector space $U$ over a field $k$ and any $a, b$ in $\Lin(U)$, we have
  $$(ab)^* = a^* \circ L_b = b^* \circ R_a$$
  where $\circ$ denotes the composition of linear maps
  $\Lin(U)\rightarrow\Lin(U)\rightarrow k$.
\end{lemma}
\begin{proof}
  For any $c$ in $\Lin(U)$, we have
  $$(ab)^*(c) = \tr(abc) = \tr(a L_b(c)) = a^*(L_b(c)) = (a^* \circ L_b)(c)$$
  and similarly
  $$(ab)^*(c) = \tr(abc) = \tr(bca) = \tr(b R_a(c)) = b^*(R_a(c)) =
  (b^* \circ R_a)(c).$$
\end{proof}

\begin{definition}
  \label{definition_t_sigma}
  For any vector space $U$, for any permutation $\sigma$ in the
  symmetric group $S_3$,
  define a map $t_\sigma$ in $\Lin(U^{\otimes3})$ by letting, for all
  $u_1, u_2, u_3$ in $U$,
  \begin{equation}
    \label{def_t_sigma}
    t_\sigma(u_1\otimes u_2\otimes u_3) = u_{\sigma(1)} \otimes
    u_{\sigma(2)} \otimes  u_{\sigma(3)}.
  \end{equation}
\end{definition}

The following lemma is classical (\cite{Markl2008}, \cite{Weyl}).

\begin{lemma}
  \label{lemma_t_sigma_star}
  For any vector space $U$, the images of $t_\id$, $t_{(1\,2\,3)}$,
  $t_{(3\,2\,1)}$ under the map $t\mapsto t^*$ from Definition
  \ref{definition_iota_star} are the following trilinear forms, given
  by their values at any $a_1 \otimes a_2 \otimes a_3$ in $\Lin(U)^{\otimes 3}$:
  \begin{eqnarray*}
    t_\id^*(a_1\otimes a_2\otimes a_3) & = & \tr(a_1) \tr(a_2) \tr(a_3), \\
    t_{(1\,2\,3)}^*(a_1\otimes a_2\otimes a_3) & = & \tr(a_1a_2a_3), \\
    t_{(3\,2\,1)}^*(a_1\otimes a_2\otimes a_3) & = & \tr(a_3a_2a_1).
  \end{eqnarray*}
\end{lemma}
\begin{proof}
  By linearity, it is enough to prove these equations in the case
  where the $a_i$ are simple tensors of the form $a_i =
  \iota(v_i\otimes\lambda_i)$ with $v_i$ in $U$ and $\lambda_i$ in $U^*
  $. Letting the dot ($\cdot$) denote multiplication in $\Lin(U^{\otimes
  3})$, for any permutation $\sigma$ in $S_3$, we have
  \begin{eqnarray*}
    t_{\sigma}^*(a_1\otimes a_2\otimes a_3) & = & \tr(t_{\sigma}
      \cdot (\iota(v_1\otimes\lambda_1) \otimes
    \iota(v_2\otimes\lambda_2) \otimes \iota(v_3\otimes\lambda_3))) \\
    & = & \tr(t_{\sigma} \cdot \iota(v_1\otimes v_2 \otimes
    v_3\otimes\lambda_1\otimes\lambda_2\otimes\lambda_3)) \\
    & = & \tr(\iota(t_{\sigma}(v_1\otimes v_2 \otimes v_3)\otimes
    \lambda_1\otimes\lambda_2\otimes\lambda_3)) \\
    & = & \tr(\iota(v_{\sigma(1)}\otimes v_{\sigma(2)} \otimes
    v_{\sigma(3)}\otimes\lambda_1\otimes\lambda_2\otimes\lambda_3)) \\
    & = & \lambda_1(v_{\sigma(1)}) \lambda_2(v_{\sigma(2)})
    \lambda_3(v_{\sigma(3)}).
  \end{eqnarray*}
  From here, the results follow for each of the three particular
  permutations $\sigma$ being considered.
\end{proof}

In the next section, in the proof of our main result (Proposition
\ref{prop_decomp_h}), we will need the following Lemma
\ref{lemma_eval_composition_with_L_t_sigma}, which is about composing
the $L_{t_\sigma}$ with forms that are simple tensors. In the proof
of Lemma \ref{lemma_eval_composition_with_L_t_sigma}, we will need
this simpler lemma about evaluating the $L_{t_\sigma}$ on simple tensors:

\begin{lemma}
  \label{lemma_eval_L}
  For any vector space $U$, for any $u_1, u_2, u_3$ in $U$, any
  $\zeta_1,\zeta_2,\zeta_3$ in $U^*$, and any permutation $\sigma$ in
  $S_3$, we have the following equality between elements of
  $\Lin(U^{\otimes 3})$:
  $$L_{t_\sigma}\Biggl(\bigotimes_{i=1,2,3} \iota(u_i \otimes
  \zeta_i)\Biggr) = \bigotimes_{i=1,2,3} \iota(u_{\sigma(i)} \otimes \zeta_i).$$
\end{lemma}
\begin{proof}
  Both sides are equal to $t_\sigma \cdot \bigotimes_{i=1,2,3}
  \iota(u_i\otimes \zeta_i).$
\end{proof}

\begin{lemma}
  \label{lemma_eval_composition_with_L_t_sigma}
  For any vector space $U$, for any $v_1, v_2, v_3$
  in $U$, any $\lambda_1,\lambda_2,\lambda_3$ in $U^*$, and any
  permutation $\sigma$ in $S_3$, we have the following equality
  between elements of $\Lin(U^{\otimes 3})^*$:
  $$
  \Biggl(\bigotimes_{i=1,2,3} \iota^*(v_i \otimes \lambda_i)\Biggr)
  \circ L_{t_\sigma} = \bigotimes_{i=1,2,3} \iota^*(v_i \otimes
  \lambda_{\sigma^{-1}(i)}).
  $$
\end{lemma}
\begin{proof}
  By linearity, it is enough to verify that both sides agree when
  evaluated on a rank one tensor of the form $\bigotimes_{i=1,2,3}
  \iota(u_i \otimes \zeta_i)$ for some $u_i$ in $U$ and
  $\zeta_i$ in $U^*$. We have:
  \begin{align*}
    & \Biggl(\Biggl(\bigotimes_{i=1,2,3} \iota^*(v_i \otimes
      \lambda_i)\Biggr) \circ
    L_{t_\sigma}\Biggr)\Biggl(\bigotimes_{i=1,2,3} \iota(u_i \otimes
    \zeta_i)\Biggr) & \\
    & = \Biggl(\bigotimes_{i=1,2,3} \iota^*(v_i \otimes
    \lambda_i)\Biggr)\Biggl(\bigotimes_{i=1,2,3} \iota(u_{\sigma(i)}
    \otimes \zeta_i)\Biggr) & \text{by Lemma \ref{lemma_eval_L}} \\
    & = \prod_{i=1,2,3} \iota^*(v_i \otimes
    \lambda_i)(\iota(u_{\sigma(i)} \otimes \zeta_i)) & \\
    & = \prod_{i=1,2,3} \lambda_i(u_{\sigma(i)}) \zeta_i(v_i) &
    \text{by Equation (\ref{iota_pairing}}) \\
    & = \prod_{i=1,2,3} \lambda_{\sigma^{-1}(i)}(u_i) \zeta_i(v_i)
    & \text{by commutativity in $k$} \\
    & = \prod_{i=1,2,3} \iota^*(v_i \otimes
    \lambda_{\sigma^{-1}(i)})(\iota(u_i \otimes \zeta_i)) &
    \text{by Equation (\ref{iota_pairing}}) \\
    & = \Biggl(\bigotimes_{i=1,2,3} \iota^*(v_i \otimes
    \lambda_{\sigma^{-1}(i)})\Biggr)\Biggl(\bigotimes_{i=1,2,3}
    \iota(u_i \otimes \zeta_i)\Biggr).
  \end{align*}
\end{proof}

\section{Main results}

Let us return to the 2-dimensional
vector space $V$ over a field $k$ that we had fixed in the
overview. Let $I$ denote the identity in $\Lin(V)$.

\begin{definition}
  \label{def_Q}
  Let $Q = \Lin(V)/kI$ denote the quotient vector space of
  $\Lin(V)$ by the scalar multiples of identity.
\end{definition}

Note that $\dim Q = (\dim V)^2 - 1 = 3$. The dual $Q^*$ is
identified with the subspace of $\Lin(V)^*$ consisting of those forms $\mu$
that satisfy $\mu(I) = 0$.

\begin{lemma}
  \label{g_volume_form}
  The trilinear form $g$ on $\Lin(V)$ is antisymmetric and passes
  to the quotient $Q$, inducing a volume form on $Q$.
\end{lemma}
\begin{proof}
  The antisymmetry follows from the definition of $g$ in Equation
  (\ref{intro_g}) and the cyclic property of the trace,
  $\tr(a_1a_2a_3) = \tr(a_2a_3a_1)$. The claim about passing to the
  quotient is that for $a_1, a_2, a_3$ in $\Lin(V)$, if any of the
  $a_i$ is a scalar multiple of identity, then
  $g(a_1, a_2, a_3) = 0$. This is verified directly, for instance
  if $a_3 = I$ then $g(a_1, a_2, a_3) = \tr(a_1a_2) - \tr(a_2a_1) =
  0$. Finally, as $\dim Q = 3$, antisymmetric 3-forms on $Q$ are
  volume forms on $Q$.
\end{proof}

\begin{lemma}
  \label{prop_decomp_g}
  For any basis $(\mu_1, \mu_2, \mu_3)$ of $Q^*$, letting
  $\varepsilon(\sigma)$ denote the signature of a permutation
  $\sigma$, there exists a scalar $\alpha$ such that
  \begin{equation}
    \label{decomp_g}
    g = \alpha \sum_{\sigma \in S_3} \varepsilon(\sigma)
    \bigotimes_{i=1,2,3} \mu_{\sigma(i)}.
  \end{equation}
\end{lemma}
\begin{proof}
  As the space of volume forms on $Q$ is one-dimensional, and by
  Lemma \ref{g_volume_form} we already know that $g$ is a volume
  form on $Q$, it is enough to check that the right-hand side is
  antisymmetric. That is true by construction, that expression
  being known as an antisymmetrized tensor product.
\end{proof}

\begin{remark}
  \label{remark_compute_alpha}
  The constant $\alpha$ in Lemma \ref{prop_decomp_g} can be
  computed by picking any $c_1, c_2, c_3$ in
  $\Lin(V)$ such that $g(c_1, c_2, c_3) = 1$ and using Equation
  (\ref{decomp_g})
  as a definition of $\alpha^{-1}$:
  \begin{equation}
    \label{expr_alpha}
    \alpha^{-1} = \sum_{\sigma\in S_3} \varepsilon(\sigma)
    \prod_{i=1,2,3} \mu_{\sigma(i)}(c_i).
  \end{equation}
\end{remark}

\begin{lemma}
  \label{lemma_g_h_t_sigma_star}
  The following equalities hold between trilinear forms on $\Lin(V)$:
  \begin{eqnarray}
    g & = & t_{(1\,2\,3)}^* - t_{(3\,2\,1)}^*, \label{g_t_sigma_star}\\
    h & = & t_\id^* - t_{(1\,2\,3)}^*. \label{h_t_sigma_star}
  \end{eqnarray}
\end{lemma}
\begin{proof}
  This follows readily from Lemma \ref{lemma_t_sigma_star} and the
  definitions of $g$ and $h$ in Equations (\ref{intro_g}, \ref{intro_h}).
\end{proof}

\begin{lemma}
  \label{relation_g_h}
  The following equality holds between forms in $\Lin(V^{\otimes3})^*$:
  $$h = g \circ L_{t_{(3\,2\,1)}}.$$
\end{lemma}
\begin{proof}
  We have
  \begin{align*}
    h & = (t_\id - t_{(1\,2\,3)})^* & \text{by Equation
    (\ref{h_t_sigma_star})} \\
    & = ((t_{(1\,2\,3)} - t_{(3\,2\,1)}) \cdot t_{(3\,2\,1)})^* & \\
    & = (t_{(1\,2\,3)} - t_{(3\,2\,1)})^* \circ L_{t_{(3\,2\,1)}} & \text{by
    Lemma \ref{star_product}} \\
    & = g \circ L_{t_{(3\,2\,1)}} & \text{by Equation
    (\ref{g_t_sigma_star}).\qedhere}
  \end{align*}
\end{proof}
While Lemma \ref{prop_decomp_g} allowed arbitrary linear forms
$\mu_i$, Proposition \ref{prop_decomp_h} will need to restrict to
\emph{rank one} forms, meaning the $\iota^*(v\otimes\lambda)$ for $v$
in $V$ and $\lambda$ in $V^*$. The necessity of that
restriction is discussed in Remark
\ref{remark_rank_one_necessary}.

\begin{lemma}
  \label{lemma_eqconv_bases}
  For $i=1,2,3$, let $v_i$ be a nonzero vector in $V$, let
  $\lambda_i$ be a nonzero linear form on $V$ such that
  $\lambda_i(v_i)=0$, and let $\mu_i = \iota^*(v_i\otimes\lambda_i)$.
  The following conditions are equivalent:
  \begin{enumerate}[nosep]
    \item \label{vnoncol} The vectors $v_1, v_2, v_3$ are pairwise
      noncolinear: $i\neq j\Rightarrow v_i\not\in\linspan(v_j)$.
    \item \label{lnoncol} The forms $\lambda_1, \lambda_2, \lambda_3$
      are pairwise noncolinear: $i\neq j\Rightarrow
      \lambda_i\not\in\linspan(\lambda_j)$.
    \item \label{muindep} The forms $\mu_1, \mu_2, \mu_3$ are
      linearly independent.
    \item \label{mubasis} The family $(\mu_1, \mu_2, \mu_3)$ is a
      basis of $Q^*$.
  \end{enumerate}
\end{lemma}
\begin{proof}
  $\ref{muindep}\Leftrightarrow\ref{mubasis}$ holds because the
  hypothesis $\lambda_i(v_i)=0$ is equivalent to $\mu_i\in Q^*$, and
  $\dim Q^*=3$. To prove the other implications, notice that for each
  $i$, we have
  $$\linspan(v_i)=\ker(\lambda_i),$$
  since $\lambda_i(v_i)=0$ means that
  $\linspan(v_i)\subset\ker(\lambda_i)$, and as $\dim V = 2$, we have
  $\dim \ker(\lambda_i)=1$ hence the inclusion is an equality. This
  means in particular that for each i, j,
  \begin{equation*}
    \text{$v_i, v_j$ are colinear $\Leftrightarrow$ $\lambda_i,
      \lambda_j$ are colinear
    $\Leftrightarrow$ $\mu_i, \mu_j$ are colinear}.
  \end{equation*}
  This readily proves the implications
  $\ref{muindep}\Rightarrow\ref{vnoncol}\Leftrightarrow\ref{lnoncol}$.
  Let us prove $\ref{vnoncol}\Rightarrow\ref{muindep}$. Suppose that
  there exists scalars $\alpha_i$ such that $\sum_i\alpha_i\mu_i=0$.
  As the $\mu_i$ are nonzero, at most one of the $\alpha_i$ can be
  zero. Thus, for some distinct indices $i,j,l$, we have
  $\alpha_i\mu_i=\alpha_j\mu_j + \alpha_l\mu_l$ with $\alpha_j\neq 0$
  and $\alpha_l\neq 0$. It follows that $\alpha_j\mu_j
  + \alpha_l\mu_l$ has rank at most one, so $\alpha_j\mu_j$ and
  $\alpha_l\mu_l$ are colinear, so $\mu_j$ and $\mu_l$ are colinear,
  so $v_j$ and $v_l$ are colinear.
\end{proof}

\begin{proposition}
  \label{prop_decomp_h}
  For any vectors $v_1, v_2, v_3$ in $V$ and linear forms $\lambda_1,
  \lambda_2, \lambda_3$ on $V$ satisfying the equivalent conditions of Lemma
  \ref{lemma_eqconv_bases}, we have:
  \begin{equation}
    \label{decomp_h}
    h = \frac{-1}{\lambda_1(v_2)\lambda_2(v_3)\lambda_3(v_1)}
    \sum_{\sigma\in S_3}
    \varepsilon(\sigma) \bigotimes_{i=1,2,3} \iota^*(v_{\sigma(i)}
    \otimes \lambda_{\sigma\cdot(1\,2\,3)(i)}).
  \end{equation}
\end{proposition}
\begin{proof}
  Let us first explain why the denominator
  $\lambda_1(v_2)\lambda_2(v_3)\lambda_3(v_1)$ is nonzero. Because of
  condition \ref{vnoncol} in Lemma \ref{lemma_eqconv_bases}, whenever
  $i\neq j$, the vector $v_j$ cannot belong to the one-dimensional
  space $\ker(\lambda_i) = \linspan(v_i)$, so $\lambda_i(v_j)\neq 0$,
  so  $\lambda_1(v_2)\lambda_2(v_3)\lambda_3(v_1)\neq 0$.

  Let us now prove Equation (\ref{decomp_h}) up to a scalar factor
  $\alpha$. Let $\mu_i = \iota^*(v_i\otimes\lambda_i)$. Lemma
  \ref{lemma_eqconv_bases} says that the $\mu_i$ form a basis of
  $Q^*$, so we can apply Lemma \ref{prop_decomp_g} with that basis to obtain
  \begin{equation*}
    g = \alpha \sum_{\sigma \in S_3} \varepsilon(\sigma)
    \bigotimes_{i=1,2,3}  \iota^*(v_{\sigma(i)}  \otimes \lambda_{\sigma(i)})
  \end{equation*}
  for some scalar $\alpha$. Lemma \ref{relation_g_h} transforms that into
  \begin{equation*}
    h = \alpha \sum_{\sigma\in S_3}
    \varepsilon(\sigma) \Biggl(\bigotimes_{i=1,2,3} \iota^*(v_{\sigma(i)}
    \otimes \lambda_{\sigma(i)})\Biggr) \circ L_{t_{(3\,2\,1)}},
  \end{equation*}
  which Lemma
  \ref{lemma_eval_composition_with_L_t_sigma} transforms into
  \begin{equation}
    \label{decomp_h_alpha}
    h = \alpha \sum_{\sigma\in S_3}
    \varepsilon(\sigma) \bigotimes_{i=1,2,3} \iota^*(v_{\sigma(i)}
    \otimes \lambda_{\sigma\cdot(1\,2\,3)(i)}).
  \end{equation}

  There only remains to evaluate the scalar $\alpha$. Let $a_i=\iota(v_i \otimes
  \lambda_i)$.
  Notice that $\tr(a_i)=\lambda_i(v_i)=0$, so
  \begin{equation}
    \label{eval_alpha_lhs}
    h(a_1, a_2, a_3) = -\tr(a_1a_2a_3) =
    -\lambda_1(v_2)\lambda_2(v_3)\lambda_3(v_1).
  \end{equation}
  On the other hand, evaluating Equation (\ref{decomp_h_alpha}) and
  simplifying that using Equation (\ref{iota_pairing}) yields
  \begin{equation}
    \label{eval_alpha_rhs}
    h(a_1, a_2, a_3) = \alpha \sum_{\sigma\in S_3}
    \varepsilon(\sigma) \prod_{i=1,2,3}
    \lambda_{\sigma\cdot(1\,2\,3)(i)}(v_i)\lambda_i(v_{\sigma(i)}).
  \end{equation}
  Since $\lambda_i(v_i)=0$, the product in Equation (\ref{eval_alpha_rhs})
  vanishes whenever $\sigma$ has a fixed point or $\sigma\cdot(1\,2\,3)$ has a
  fixed point. Thus the only $\sigma$ contributing to the sum is
  $\sigma=(1\,2\,3)$. Thus, Equation (\ref{eval_alpha_rhs})
  simplifies to
  \begin{equation}
    \label{eval_alpha_rhs_123}
    h(a_1, a_2, a_3) = \alpha  \prod_{i=1,2,3}
    \lambda_{(3\,2\,1)(i)}(v_i)\lambda_i(v_{(1\,2\,3)(i)}).
  \end{equation}
  further simplifying as
  $$h(a_1, a_2, a_3) = \alpha (\lambda_1(v_2)\lambda_2(v_3)\lambda_3(v_1))^2.$$
  Combining that with Equation (\ref{eval_alpha_lhs}) yields
  $$\alpha=\frac{-1}{\lambda_1(v_2)\lambda_2(v_3)\lambda_3(v_1)}\cdot$$
\end{proof}

\begin{remark}
  \label{remark_rank_one_necessary}
  Two tensors $p,q$ in $\Lin(V)^{*\otimes 3}$ related to each other
  in the same way as $g$ and $h$ are related by Lemma
  \ref{relation_g_h}, namely $q = p \circ L_{t_{(3\,2\,1)}}$, may still
  fail to have the same tensor rank if their tensor decompositions
  involve linear terms in $\Lin(V)^*$ that are not of rank one.
\end{remark}
\begin{proof}
  Consider the counterexample of $p = t_{(1\,2\,3)}^*$ and $q =
  t_{\id}^*$. The same argument as in the proof of Lemma
  \ref{relation_g_h} yields $q = p \circ L_{t_{(3\,2\,1)}}$. As noted in
  Lemma \ref{lemma_t_sigma_star}, for $a_1, a_2, a_3$ in $\Lin(V)$,
  we have $p(a_1, a_2, a_3) = \tr(a_1a_2a_3)$ and
  $q(a_1a_2a_3)=\tr(a_1)\tr(a_2)\tr(a_3)$. Thus, as tensors in
  $\Lin(V)^{*\otimes 3}$, $q$ has rank one but $p$ does not.
\end{proof}

To elaborate on the previous remark, the linear form $a\mapsto\tr(a)$
does not have rank one, so even though $q$ has rank one as a tensor
of order 3 in $\Lin(V)^{*\otimes 3}$, it does not have rank one as a
tensor of order 6 in $(V\otimes V^*)^{\otimes 3}$, and our tool for
transporting tensor decompositions, Lemma
\ref{lemma_eval_composition_with_L_t_sigma}, applies to tensors of
order 6 in $(V\otimes V^*)^{\otimes 3}$.

\section{Strassen algorithms}

Proposition \ref{prop_decomp_h} is already a form of Strassen's
algorithm, but that may be obscured by the tensor formalism, so let us derive a
few more concrete statements as corollaries.

\begin{corollary}
  \label{concrete_decomp_h}
  For any vectors $v_1, v_2, v_3$ in $V$ and linear forms $\lambda_1,
  \lambda_2, \lambda_3$ on $V$ satisfying the equivalent conditions of Lemma
  \ref{lemma_eqconv_bases}, for all $a_1, a_2, a_3$ in
  $\Lin(V)$,
  \begin{multline*}
    \tr(a_1a_2a_3)  = \tr(a_1)\tr(a_2)\tr(a_3) \\
    + \frac{1}{\lambda_1(v_2)\lambda_2(v_3)\lambda_3(v_1)} \sum_{\sigma\in S_3}
    \varepsilon(\sigma) \prod_{i=1,2,3}
    \lambda_{\sigma\cdot(1\,2\,3)(i)}(a_i(v_{\sigma(i)})).
  \end{multline*}
\end{corollary}
\begin{proof}
  Evaluating Equation (\ref{decomp_h}) at any $a_1,
  a_2, a_3$ in $\Lin(V)$ gives:
  \begin{multline*}
    \tr(a_1)\tr(a_2)\tr(a_3) - \tr(a_1a_2a_3) =
    \\ \frac{-1}{\lambda_1(v_2)\lambda_2(v_3)\lambda_3(v_1)}
    \sum_{\sigma\in S_3}
    \varepsilon(\sigma) \prod_{i=1,2,3} \iota^*(v_{\sigma(i)}
    \otimes \lambda_{\sigma\cdot(1\,2\,3)(i)})(a_i)
  \end{multline*}
  and the result follows by Definition \ref{definition_iota_star}.
\end{proof}

\begin{corollary}
  \label{concrete_decomp_h_bilinear}
  For any vectors $v_1, v_2, v_3$ in $V$ and linear forms $\lambda_1,
  \lambda_2, \lambda_3$ on $V$ satisfying the equivalent conditions of Lemma
  \ref{lemma_eqconv_bases}, for all $a_1, a_2$ in
  $\Lin(V)$,
  \begin{multline}
    \label{eq_concrete_decomp_h_bilinear}
    a_1a_2 = \tr(a_1)\tr(a_2)I \\ +
    \frac{1}{\lambda_1(v_2)\lambda_2(v_3)\lambda_3(v_1)} \sum_{\sigma\in S_3}
    \varepsilon(\sigma)
    \tr(a_1c_{\sigma(1), \sigma(2)})
    \tr(a_2c_{\sigma(2), \sigma(3)})
    c_{\sigma(3), \sigma(1)}
  \end{multline}
  where $c_{i,j}$ in $\Lin(V)$ is defined by $c_{i,j}(u) = \lambda_j(u)
  v_i$ for all $u$ in $V$.
\end{corollary}
\begin{proof}
  Let $x$ denote the right-hand side of Equation
  (\ref{eq_concrete_decomp_h_bilinear}). The claim is that $a_1a_2 =
  x$. That is equivalent to the claim that $\tr(a_1a_2a_3) = \tr(xa_3)$
  for all $a_3$ in $\Lin(V)$. That claim is directly verified by
  comparing the expression of $\tr(a_1a_2a_3)$ given by Corollary
  \ref{concrete_decomp_h} to the expression of $\tr(xa_3)$ expanded by
  using the definition of $x$,
  noting that $c_{i,j}=\iota(v_i\otimes\lambda_j)$.
\end{proof}

\begin{corollary}
  \label{good_old_strassen}
  The original Strassen algorithm is obtained by applying Corollary
  \ref{concrete_decomp_h_bilinear} to the vector space $V=k^2$, with
  the following choices: $v_1=\left(
    \begin{smallmatrix}1\\0
  \end{smallmatrix}\right)$, $\lambda_1=\left(
    \begin{smallmatrix}0 & 1
  \end{smallmatrix}\right)$, $v_2=\left(
    \begin{smallmatrix}0\\1
  \end{smallmatrix}\right)$, $\lambda_2=\left(
    \begin{smallmatrix}1 & 0
  \end{smallmatrix}\right)$,
  $v_3=\left(
    \begin{smallmatrix}1\\1
  \end{smallmatrix}\right)$, $\lambda_3=\left(
    \begin{smallmatrix}1 & -1
  \end{smallmatrix}\right)$.
\end{corollary}
\begin{proof}
  Applying Corollary \ref{concrete_decomp_h_bilinear},
  expanding the sum over all 6 permutations, and noticing that
  $\lambda_1(v_2)\lambda_2(v_3)\lambda_3(v_1) = 1$,
  we obtain the following
  matrix multiplication algorithm: for any two $2\times 2$ matrices $a, b$,
  \begin{align*}
    ab & = \tr(a) \tr(b) I  \\
    & + \tr(ac_{1,2}) \tr(bc_{2,3}) c_{3,1} \\
    & + \tr(ac_{2,3}) \tr(bc_{3,1}) c_{1,2} \\
    & + \tr(ac_{3,1}) \tr(bc_{1,2}) c_{2,3} \\
    & - \tr(ac_{2,1}) \tr(bc_{1,3}) c_{3,2} \\
    & - \tr(ac_{1,3}) \tr(bc_{3,2}) c_{2,1} \\
    & - \tr(ac_{3,2}) \tr(bc_{2,1}) c_{1,3}
  \end{align*}
  where the $c_{i,j}=\iota(v_i \otimes \lambda_j)=v_i\lambda_j$ are:
  \begin{eqnarray*}
    c_{1,2} = v_1\lambda_2 = \biggl(
      \begin{matrix}1 & 0 \\ 0 & 0
    \end{matrix}\biggr), &
    c_{1,3} = v_1\lambda_3 = \biggl(
      \begin{matrix}1 & -1 \\ 0 & \phantom{-}0
    \end{matrix}\biggr), \\
    c_{2,3} = v_2\lambda_3 = \biggl(
      \begin{matrix}0 & \phantom{-}0 \\ 1 & -1
    \end{matrix}\biggr), &
    c_{2,1} = v_2\lambda_1 = \biggl(
      \begin{matrix}0 & 0 \\ 0 & 1
    \end{matrix}\biggr), \\
    c_{3,1} = v_3\lambda_1 = \biggl(
      \begin{matrix}0 & 1 \\ 0 & 1
    \end{matrix}\biggr), &
    c_{3,2} = v_3\lambda_2 = \biggl(
      \begin{matrix}1 & 0 \\ 1 & 0
    \end{matrix}\biggr).
  \end{eqnarray*}
  Let $a^{i,j}$ and $b^{i,j}$ denote the matrix coefficients, using
  superscript notation to distinguish that from the subscripts used
  to index the $c_{i,j}$ matrices. Let $e_{i,j}$ be the elementary
  matrix with a 1 at position $(i,j)$ and zeros elsewhere. Using the
  above table of $c_{i,j}$ matrices, the above equation expands to
  \begin{align*}
    ab & = (a^{1,1} + a^{2,2})(b^{1,1} + b^{2,2})(e_{1,1} + e_{2,2})  \\
    & + a^{1,1} (b^{1,2}-b^{2,2}) (e_{1,2} + e_{2,2}) \\
    & + (a^{1,2}-a^{2,2})(b^{2,1} + b^{2,2}) e_{1,1} \\
    & + (a^{2,1} + a^{2,2}) b^{1,1} (e_{2,1} - e_{2,2}) \\
    & - a^{2,2} (b^{1,1} - b^{2,1}) (e_{1,1} + e_{2,1}) \\
    & - (a^{1,1} - a^{2,1}) (b^{1,1} + b^{1,2}) e_{2,2} \\
    & - (a^{1,1} + a^{1,2}) b^{2,2} (e_{1,1} - e_{1,2}).
  \end{align*}
  These bilinear forms in the $a^{i,j}$ and $b^{i,j}$ are exactly the
  terms I, II, III, IV, V, VI, VII introduced in the original
  Strassen article \cite{Strassen}:
  \begin{align*}
    ab & = \mathrm{I}\cdot(e_{1,1} + e_{2,2})  \\
    & + \mathrm{III}\cdot (e_{1,2} + e_{2,2}) \\
    & + \mathrm{VII}\cdot e_{1,1} \\
    & + \mathrm{II}\cdot (e_{2,1} - e_{2,2}) \\
    & + \mathrm{IV}\cdot(e_{1,1} + e_{2,1}) \\
    & + \mathrm{VI}\cdot e_{2,2} \\
    & + \mathrm{V}\cdot (e_{1,2} - e_{1,1}).
  \end{align*}
  Thus the coefficients of the product matrix $ab$ are:
  \begin{align*}
    (ab)^{1,1} & = \mathrm{I} + \mathrm{IV} - \mathrm{V} + \mathrm{VII} \\
    (ab)^{1,2} & = \mathrm{III} + \mathrm{V} \\
    (ab)^{2,1} & = \mathrm{II} + \mathrm{IV} \\
    (ab)^{2,2} & = \mathrm{I} - \mathrm{II} + \mathrm{III} + \mathrm{VI}
  \end{align*}
  exactly as originally stated by Strassen \cite{Strassen}.
\end{proof}

\bibliography{main}{}
\bibliographystyle{IEEEtran}

\end{document}